\RequirePackage{fix-cm}
\documentclass[smallextended]{svjour3}       % onecolumn (second format)
\smartqed  % flush right qed marks, e.g. at end of proof
\usepackage{graphicx}
\usepackage{mathptmx}      % use Times fonts if available on your TeX system
\usepackage{amsfonts, amssymb, amsmath,nicefrac}
\newcommand{\bq}{\begin{equation}}
\newcommand{\eq}{\end{equation}}
\newcommand{\be}{\begin{eqnarray*}}
\newcommand{\ee}{\end{eqnarray*}}
\newcommand{\ben}{\begin{eqnarray}}
\newcommand{\een}{\end{eqnarray}}

\newcommand{\lb}{\left}
\newcommand{\rb}{\right}

\newcommand{\bPhi}{\overline{\Phi}}
\newcommand{\bD}{\overline{D}}
\begin{document}

\title{Portfolio optimization in the case of an asset with a given liquidation time distribution.%\thanks{Grants or other notes
}
\subtitle{}


\author{Ljudmila A. Bordag $^{\ast}$       \and
       	Ivan P. Yamshchikov $^{\ast}$ \thanks{$^\ast$ This research was supported by the European Union in the FP7-PEOPLE-2012-ITN Program under Grant Agreement Number 304617 (FP7 Marie Curie Action, Project Multi-ITN STRIKE - Novel Methods in Computational Finance).\vspace{6pt}}	\and
        Dmitry Zhelezov
}


\institute{Ljudmila A. Bordag  \at
            University of Applied Sciences Zittau/Goerlitz, Germany \\
            D-02763, Zittau, Theodor-Koerner-Allee, 16\\
              Tel.: +49-3583 61-1488\\
              Fax: +49-3583 61-1262\\
              \email{l.bordag@hszg.de}
           \and
           	Ivan P. Yamshchikov \at
              University of Applied Sciences Zittau/Goerlitz, Germany \\
              \email{i.yamshchikov@hszg.de}
            \and
            Dmitry Zhelezov \at
             Mathematical Sciences, Chalmers, 41296 Gothenburg, Sweden; \\
             Mathematical Sciences, University of Gothenburg, 41296 Gothenburg, Sweden\\
             Tel: +46 31 772 35 82\\
             \email{zhelezov@chalmers.se}
}

\date{Received: date / Accepted: date}

\maketitle

\begin{abstract}
Management of the portfolios containing low liquidity assets is a tedious problem. The buyer proposes the price that can differ greatly from the paper value estimated by the seller, the seller, on the other hand, can not liquidate his portfolio instantly and waits for a more favorable offer. To minimize losses in this case we need to develop new methods. One of the steps moving the theory towards practical needs is to take into account the time lag of the liquidation of an illiquid asset. This task became especially significant for the practitioners in the time of the global financial crises.
Working in the Merton's optimal consumption framework with continuous time we consider an optimization problem for a portfolio with an illiquid, a risky and a risk-free asset. While a standard Black-Scholes market describes the liquid part of the investment the illiquid asset is sold at a random moment with prescribed liquidation time distribution. In the moment of liquidation it generates additional liquid wealth dependent on illiquid assetÕs paper value. The investor has the logarithmic utility function as a limit case of a HARA-type utility. Different distributions of the liquidation time of the illiquid asset are under consideration - a classical exponential distribution and Weibull distribution that is more practically relevant. Under certain conditions we show the existence of the viscosity solution in both cases. Applying numerical methods we compare classical MertonÕs strategies and the optimal consumption-allocation strategies for portfolios with different liquidation-time distributions of an illiquid asset.
\end{abstract}

\section{Introduction}

Since the last financial crises in 2008 the problems connected with development of optimal strategies  for portfolios with illiquid assets are becoming more and more important for large market participants. Many financial institutes in developed countries have to solve practical problems connected with a liquidation of the assets with a very low liquidity or problems with a management of such portfolios especially if such illiquid assets provide stochastic incomes or down payments (like taxes or other obligations).

There exist a reasonable amount of papers devoted  to the problem of the optimal consumption and liquidity. In general, the most challenging task one faces defining a model is to incorporate the illiquidity in a mathematically tractable way. Intuitively it is clear which of the assets we would call liquid. Majority of the researchers would say that these are assets which can be well modeled by the Black-Scholes model.Yet there is still no widely accepted way of defining illiquidity of an asset as a comparable and measurable parameter. It is also important to note that the mathematically correct definition, being a problem itself, is not the biggest challenge in the tasks of such kind. In fact, the exact formulation of the goals of the portfolio optimization is even more tedious. Moreover, the illiquidity is usually connected with the strong bid-ask difference and with an essential lag-time while liquidating such positions. Stochastic processes that describe such effects are not studied profoundly in financial mathematics. All these factors make the problem really interesting and complicated, so we would like to give a brief overview of the models and problems that are relevant to our research.

 \subsection{State-of-the-art.}
 Here we would mention only a few results most relevant to our case of an optimal allocation-consumption problem for a portfolio with an illiquid asset and random income.

In 1993 in order to find the optimal policies {\em Duffie and Zariphopoulou} in \cite{ZaripDuf} develop the framework of the optimal consumption for the continuous time model, proposed by {\em Merton}, \cite{Merton}. They considered an infinite time horizon proved the existence and uniqueness of the viscosity solution of the associated HJB equation for the class of concave utility functions. They posed the following conditions on the utility function $U(c)$
	\begin{eqnarray} \label{A}
 U ~~{\rm in~~} c {~~is~~ strictly~~ concave},~~~C^2(0, +\infty),\nonumber\\	
U(c) \leq M(1 + c)^{\gamma},~~ {\rm with}~~0 < \gamma < 1, M  > 0,\nonumber\\
U(0) \geq 0, \,\,\, {\displaystyle \lim_{c \to 0} U'(c} = +\infty,	\,\,\, \displaystyle \lim_{c \to \infty} U'(c) = 0.\nonumber
\end{eqnarray}
Some years later, in 1997, in \cite{ZaripDufFlem} an extended problem of hedging in incomplete markets with hyperbolic absolute risk aversion (so called HARA) utility function was studied. Here the stochastic income cannot be replicated by trading the available securities. The economical setting was of the following type:
\begin{itemize}
\item The investor receives stochastic income at time moment $t$ at the rate $Y_t$, where $ d Y_t=\mu Y_t +  \eta Y_t d W_t^1, ~~ t\ge 0, Y_0=y, ~~ y\ge 0 $ and $\mu, \eta>0 - const$ here $W^1$ is a standard Brownian motion.
\item The riskless bank account has a constant continuously compound interest rate $r$.
\item A traded security has a price $S$ given by $d S_t=\alpha S_t +  \sigma S_t  (\rho d W_t^1 +\sqrt{1-\rho^2} d W_t^2 )$, $\alpha, ~\sigma>0 - const$  and $W^2$ is an independent standard Brownian motion, $\rho \in (-1,1)$ is a correlation between price process $S_t$ and $Y_t$.
\item The investor utility function for consumption process $c_t$ is given by
\begin{equation}
{\cal U}(c(t))=E\left[{\displaystyle \int_0^{\infty}{e^{-\kappa t} U(c(t)) d t}}\right] \label{ZaripDufFlem},~~U(c(t))=c(t)^{\gamma},
\end{equation}
where $\gamma \in (0,1)$ and $\kappa$ is an discount factor $\kappa >r$.
\item The investors wealth process $L$ evolves $$d L_t=[r L_t  + (\alpha+\delta-r) \pi_t - c_t +Y_t] dt +\sigma \pi_t  (\rho d W_t^1 +\sqrt{1-\rho^2}~ d W_t^2 ), ~~ t\ge 0, L_0=l,$$
         where $\delta$ could be regarded as the dividends payed constantly from an illiquid asset or as the possession costs, the value $l$ is the initial wealth endowment and $\pi_t$ represents the investment in the risky asset $S$, with the remaining wealth held in riskless borrowing or lending. The goal is to characterize the investor value function $$ V(l,y)= {\displaystyle \sup_{(\pi, c) \in {\cal A}(l,y)} }{\cal U}(C).$$
         The set ${\cal A}(l,y)$ is a set of admissible controls $(\pi, c)$ such that $L_t\ge 0$.
\end{itemize}
\begin{remark}
{\em The notation of the strategy $(\pi, c)$ is standard for the problems of such kind. Throughout this paper we will denote the amount of the investment in a liquid risky asset as $\pi$ and investor's consumption as $c$. Both controls do depend on time, so to emphasize it to the reader we might also use $(\pi(t), c(t))$ or even $(\pi_t, c_t)$  from time to time.}
\end{remark}

The authors in \cite{ZaripDufFlem} proved the smoothness of the viscosity solution of the associated Hamilton-Jacobi-Bellman (HJB) equation in the  case of the HARA utility function and the infinite time horizon. This proof heavily relies on a reduction of the initial HJB equation to an ODE. After this reduction the main result follows from the uniform convergence of the classical solution of a uniformly elliptic equation to the viscosity solution, which is unique.

It is important to mention that the authors use the discount factor $e^{-\kappa t}$ in (\ref{ZaripDufFlem}) as a technical factor which is not related to stochastic income. The economical setting does not imply any liquidation of an illiquid asset which provides stochastic income $Y_t$.

 {\em Schwarz and Tebaldi} in \cite{tebaldi} used a model of random income extensively studied before, but interpreted it in an original way. They assumed that the non-traded illiquid asset generates a flow of random income in the form of dividends, until it is sold at a fixed moment of time. The authors define illiquid asset as an asset that can not be sold neither piece by piece nor at once before the investment's horizon, denoted as $T$, which is a fixed deterministic value at which the asset generates a random cash-flow equal to its' paper-value at this moment $T$ (the cash-flow is denoted as $H_T$). This model is generally related to the model of stochastic income studied by Merton in \cite{Merton} (1971), who studied the case of HARA, logarithmic and exponential utility functions and both finite and infinite time horizons. However, since the problem gets another economical reasoning behind it and becomes a sensible model of illiquidity when formulated in this way, it arises high interest and needs a more exact qualitative and quantitive description. Moreover, it gives an intuition of possible extensions, for example, incorporation of other types of utility functions or weakening the trading conditions for an illiquid asset. It all makes this model extremely interesting, so let us describe it in detail.

\subsection{Economical setting of the problem posed by Schwarz and Tebaldi}
	\label{setting1}
A risk-free bank account $B_t$ with the interest rate $r$
	\begin{equation} \label{bond_r}
		dB_t = rB_tdt, \, t \leq T,
	\end{equation}
	where $r$ is assumed to be constant.
 A stock price $S_t$, which follows the geometrical Brownian motion
	\begin{equation} \label{asset_S}
		dS_t = S_t(\alpha dt + \sigma dW^{1}_t), \, t \leq T,
	\end{equation}
	with the continuously compounded rate of return $\alpha > r$ and the standard deviation $\sigma$, $r, \alpha, \sigma - const$.
 An illiquid asset $H_t$ that can not be traded up to the time $T$ and which paper value is correlated with the stock price and follows
	\begin{equation}
		\frac{dH_t}{H_t} = (\mu - \delta)dt + \eta(\rho dW^{1} + \sqrt{1 - \rho^2}dW^{2}), \, t \le T.
	\label{eq:H1}
	\end{equation}
	where $\mu$ is the expected rate of return of the risky illiquid asset, $(W^{1},W^{2})$ are two independent standard Brownian motions, $\delta$ is the rate of dividend
paid by the illiquid asset, $\eta$ is the continuous standard deviation of the rate of return, and $\rho \in(-1; 1)$ is the correlation coefficient between the stock index and the illiquid risky asset. The parameters $\mu$,  $\delta$, $\eta$, $\rho$ are all assumed to be constant, as well as $T$, that is a fixed liquidation time.\\

Authors in \cite{tebaldi} assume that the consumption stream $c(t)$ is admissible if and only if it is positive and there exists a strategy that finances it. All the income is derived from the capital gains and the investor must be solvent. In other words, the liquid wealth process $L_t$ must cover the consumption stream, i.e satisfy the balance equation
\begin{equation}
	 dL_t = (rL_t + \delta H_t - c(t)) dt + \pi(t)\left(\frac{dS_t}{S_t} - r dt \right) + {\sf \Delta}(t - T)H_T dt,~~L_t=l,
  \label{eq:L}
\end{equation}
where the term ${\sf \Delta}(t - T)H_T$ corresponds to the instantaneous liquidation of the illiquid asset in the final time $T$. We use here ${\sf \Delta}$ for a Dirac delta-functional that could be understood as a functional that acts as $ {\sf \Delta}: f(t) \rightarrow f(0)$. This term makes the liquid wealth function discontinuous in the point $T$ when we instantly transform our illiquid asset that had only a paper value before into a sum of liquid capital. Throughout the paper we mostly use the notation from \cite{tebaldi} so that $(\Omega, \mathcal{G}, \mathbb{P})$  denotes a fixed complete probability space and $(\mathcal{G}_t)_{t \geqslant 0}$ is a given filtration. The filtration $\mathcal{G}_t $ is generated by the Brownian motion $W = (W^1, W^2)$.
The optimal portfolio allocation means that the investor wants to maximize the utility of the consumption stream. Because the market is incomplete for the given investor's utility function $U$ the optimal portfolio allocation could be described with the following functional
\begin{equation}
\label{eq1}
	\mathcal{U}(t, c, W^{\pi, c}_T) := E^{\mathbb{P}}_t\left[\int^T_t{e^{-\kappa \tau}U(c(\tau))}d\tau + \beta e^{-\kappa T}U(W^{\pi, c}_T)\right],
\end{equation}
where $W^{\pi, c}_T$ is the total wealth of the investor that uses a strategy $(\pi, c)$ till the moment $T$, parameters $\kappa, \beta$ are constants, and $E^{\mathbb{P}}_t$ denotes a mathematical expectation in terms of $(\Omega, \mathcal{G}, \mathbb{P})$ introduced above. It is important to note that $e^{-\kappa \tau}$ is a discount factor  that could be considered as a technical factor and was typically introduced in the infinite-time cases. \\
The authors studied the problem of choosing the investment and consumption policies $(\pi, c)$ so as to maximize the expected utility when wealth evolves according to the equation (\ref{eq:L}) with $(\pi, c) \in \mathcal{A}(t, l, h) $ where $\mathcal{A}(t, l, h)$ denotes the set of all admissible investment and consumption plans with initial capital $l$ and starting paper value of the illiquid asset equal to $h$.
 The investment problem then is to find such $(\pi_t, c_t) \in \mathcal{A}(t, l, h)$ that the value function of the portfolio $V(t,l,h)$ will be
\begin{equation} \label{eq:sup}
  V(t,l,h)= \sup_{(\pi , c) \in \mathcal{A}(t, l, h)} \mathcal{U}(t, c, W^{\pi, c}_T),
\end{equation}
The authors are the first that we know of who obtain an analytical solution of this problem in the form of an infinite series in the case of time separable power utility of consumption and terminal wealth.
However, as we know, the case of the logarithmic utility was not fully covered up to now. We will discuss it later in Section \ref{expsec} of this paper.

One of the possible extensions of this problem was done by {\em Ang, Papanikolaou and Westerfeld} in \cite{Papanikolaou}. They considered exactly the same model as in  \cite{tebaldi}. However, they assumed that an illiquid asset can be traded but only at infrequent, stochastic moments of time and thus the whole three-asset portfolio could be rebalanced. With a series of numerical calculations they provide an intuition of the influence of illiquidity on the marginal utility of the investor. The authors numerically study the cases when amount of the illiquid wealth is significantly bigger than the amount of the liquid capital and comparing it with the opposite case (insignificantly small amount of illiquid wealth) they show that the effects of the asset being illiquid may cause unbounded deviations from the Merton solution.

In 2008 {\em He} \cite{He} proposed a model with the same set-up but different constraints on illiquid asset. While the investor can instantaneously transfer funds from the liquid to the illiquid asset, the vice versa transaction is allowed only in exponentially distributed moments of time. The author finds an approximate numerical solution of the problem for the constant risk-aversion (CARA) utility function.

In our work we develop a connection between the model of illiquidy \cite{tebaldi} and the optimal consumption problem with an undiversifiable future income. We substitute the somewhat artificial constraint (which is less probable in practice) that the liquidation time  $T$ is fixed from very beginning with the assumption that it is stochastically distributed.
Moreover, we formulate the problem in a general case with an arbitrary liquidation time distribution and some utility functions in the next Section. Then in the Section \ref{expsec} and Section \ref{weibullsec} we investigate two special cases with logarithmic utility function and two liquidation time distributions: exponential one (which was partly studied in the works mentioned above) and the Weibull-distribution (which was used in this framework in \cite{boyazh} for the first time to our knowledge). One of the important contributions of this paper is that using the technique of the viscosity solutions we show the existence and uniqueness of the solution of the HJB equation that corresponds to the case of Weibull-distribution which, in turn, justifies further application of numerical techniques to this problem. Generally speaking, the HJB equation can be degenerate and the value function does not need not be smooth. This could happen due to the market incompleteness, imperfect correlation between the liquid and illiquid asset and a stochasticity of the income. Though this is not the case of this particular article, even in this situation one would like to get a numerical approximation for the value function and optimal policies. Strong stability of the viscosity solutions allows to get such approximations with a range of monotone and consistent schemes (for example, as it was done by {\em Munk} in \cite{Munk}).
\begin{remark}
{\em In \cite{boyazh} we, to our knowledge, for the first time proposed to study an discounting other then exponential, though the exponential case has been attracting a lot of attention. In this paper we clarify and deepen the ideas mentioned in \cite{boyazh} we also provide detailed proofs of all theorems.}
\end{remark}
In Section \ref{problem} we describe our problem in general case. We assume that the utility function is of the HARA type and the liquidation time is a random variable with some typical distribution.

\section{General Problem} \label{problem}
\subsection{Economical setting}
We assume that the investor's portfolio includes a riskless bond, a risky asset and a non-traded asset that generates stochastic income i.e. dividends. The mathematical model is defined similarly to \cite{tebaldi}. We use  $B_t, S_t, H_t$ as described in the previous Section in (\ref{bond_r})~--~(\ref{eq:H1}). However, we replace the time $T$ that was fixed before with a stochastically distributed time $\tau$.

The liquidation time $\tau$ is now a random-distributed continuous variable which does not depend on the Brownian motions $(W^1,W^2)$. The probability density function of $\tau$ distribution is denoted by $\phi(t)$ whereas $\Phi(t)$ denotes the cumulative distribution function, and  $\overline{\Phi}(t)$ the survival function also known as {\em a reliability function $\overline{\Phi}(t) = 1 - \Phi (t)$}. We omit here the explicit notion of the possible parameters of distribution in order to make the formulae shorter.

Given the filtration $\{\mathcal{F}_t \}$ generated by the Brownian motion $W = (W^1, W^2)$ we assume that the consumption process is
an element of the space $\mathcal{L}_+$ of non-negative $\{\mathcal{F}_t \}$-progressively measurable processes
$c_t$ such that
\begin{equation} \label{integrability}
E \left( \int_0^s c(t) dt \right) < \infty, ~~ s \in [0, \tau].
\end{equation}

The investor wants to maximize the average utility consumed up to the time of liquidation, given by
\begin{equation}
	\mathcal{U}(c) := E \left[\int^\tau_0 U(c(t)) dt \right]. \label{eq:sup21}
\end{equation}
Though Merton in \cite{Merton} describes the most general case of the problem and introduces a utility function which depends on $c(t)$ and $t$ in this paper we focus on $U(c(t))$ that does not depend on time explicitly.
The wealth process $L_t$ is the sum of cash holdings in bonds, stocks and {\em random} dividends from the non-traded asset minus the consumption stream. Thus, we can write
\begin{eqnarray}
	 dL_t &=& (rL_t + \delta H_t + \pi_t(\alpha - r) - c_t) dt + \pi_t\sigma dW^1_t. \nonumber
\end{eqnarray}
  The set of admissible policies is standard and consists of investment strategies $(\pi_t, c_t)$ such that
\begin{enumerate}
  \label{intergability}
 	\item $c_t$ belongs to $\mathcal{L}_+$,
	\item $\pi_t$ is $\{\mathcal{F}_t \}$-progressively measurable and $\int^s_t (\pi_\tau)^2 d\tau < \infty$ a.s. for any $t \leq s \leq \tau $,
	\item $L_\tau$, defined by the stochastic differential equation (\ref{eq:L}) and initial conditions $L_t = l > 0, \, H_t = h > 0$ a.e. ($t \leq \tau $).
\end{enumerate}

We claim that one can explicitly average (\ref{eq:sup21}) over $\tau$ and with the certain conditions posed on $\bPhi$ and $U(c)$ the problem (\ref{eq:sup21}) is equivalent to the maximization of
\begin{equation}
	\mathcal{U}(c) := E \left[\int^\infty_0{\bPhi(t)U(c(t))}dt \right]. \label{eq:supinfty}
\end{equation}
\begin{remark}
{\em It is important to note, that if $\tau$ is exponentially distributed we get \emph{precisely} the problem of optimal consumption with random income that was studied in \cite{ZaripDufFlem} and already discussed in introduction.}
\end{remark}
We demonstrate here a formal derivation of the equivalence between two optimal problems briefly mentioned by Merton in \cite{Merton}.
\begin{proposition}\label{proposition}
The problems (\ref{eq:sup21}) and (\ref{eq:supinfty}) are equivalent provided
\bq
	\lim_{T \to \infty}\bPhi(T)E\lb[U(c(T)\rb] = 0. \label{eq:utilitycondition}
\eq
\end{proposition}

\begin{proof}
	We have
	\ben
		E \left[\int_{0}^{\tau} U(c(t)) dt\right] &=& \int_{0}^{\infty} \phi(\tau) \mathcal{E} \lb[\int_{0}^{\tau} U(c(t)) dt\rb] d\tau  \nonumber \\
																	 &=& 	\int_{0}^{\infty} \int_{0}^{\tau} \phi(\tau) g(t) d\tau dt, \label{eq:change}
	\een
	where $g(t) = \mathcal{E} [U(c(t)]$. Here we used $\mathcal{E}$ to indicate that we are averaging over all random variables excluding $\tau$. Because of the absolute convergence
	$$
		\mathcal{E} \lb[\int_{0}^{\tau} U(c(t)) dt\rb] = \int_{0}^{\tau} g(t) dt
	$$
	and integrating (\ref{eq:change}) by parts we get
	\ben \label{generaleqn}
		\int_{0}^{\infty} \int_{0}^{\tau} \phi(\tau) g(t) d\tau dt &=& \bPhi(\tau) \int_0^{\tau} g(t) dt \big|_0^{\infty} \nonumber \\
																															 &+& \int_{0}^{\infty} \bPhi(t) g(t) dt = \mathcal{E} \lb[\int_{0}^{\infty} \bPhi(t) U(c(t))dt\rb],\nonumber
	\een
	where we used the condition  (\ref{eq:utilitycondition}) to eliminate the first term, and the absolute convergence of the integral to move
	the expectation out. $\bullet$
\end{proof}

\begin{remark} {\em It is interesting to note that although the investor receives additional wealth from the liquidation of the illiquid asset, the expected value of this lump payment is just a constant and does not affect the strategy for maximizing (\ref{eq:sup}).
Indeed, if we look for a supremum of $E \left[\int_{0}^{\tau} U(c(t)) + U(W^{\pi,c}_{\tau})dt\right]$ instead to a supremum of $E \left[\int_{0}^{\tau} U(c(t)) dt\right]$, then after the calculations similar to the ones in the proof of Proposition \ref{proposition} we obtain
$$E \left[\int_{0}^{\tau} U(c, t) + U(W_{\tau})dt\right]=\mathcal{E} \lb[\int_{0}^{\infty} \bPhi(t) U(c(t))dt\rb]+\int_{0}^{\infty} \phi(\tau) U(W_{\tau})dt.$$
The second term turns out to be a constant and can not change the optimal strategy.}
\end{remark}
From now on in this paper we will work with the problem (\ref{eq:sup21}) with random liquidation time $\tau$ that has a distribution satisfying the condition (\ref{eq:utilitycondition}) in Proposition \ref{proposition} and, therefore, corresponds to the \emph{value function} $V(t, l, h)$ which is defined as
\begin{equation} \label{valueFun}
	V(t, l, h) = \max_{(\pi, c) } E \left[ \int_t^\infty \bPhi(t)U(c(t)) dt | L(t) = l, H(t) = h \right].
\end{equation}

{\em Schwartz and Tebaldi} in \cite{tebaldi} write the HJB equation for the value function in terms of $l$ and $h$. We use the same notation and derive a more general  HJB equation that is the main equation in this paper
\begin{eqnarray}
V_t (t, l, h) &+& \frac{1}{2}\eta^2h^2V_{hh} (t, l, h) + (rl + h)V_l (t, l, h) \nonumber \\
  &+& (\mu - \delta) hV_h (t, l, h) + \max_{\pi} G[\pi] + \max_{c \geq 0} H[c] = 0, \label{eq:HJB21}\\
G[\pi] &=& \frac{1}{2}V_{ll}(t, l, h)\pi^2 \sigma^2 + V_{lh}(t, l, h)\eta\rho\pi\sigma h \nonumber \\
      &+& \pi(\alpha - r)V_l(t, l, h), \label{eq:Gmax21} \\
H[c] &=& -cV_l (t, l, h) + \bPhi(t) U(c), \label{eq:Hmax21}
\end{eqnarray}
with the boundary condition
\be
V(t, l, h) \to 0, \text{ as } t \to \infty.
\ee
\begin{remark} \label{remark}
{\em Merton formulates the problem in terms of the total wealth $W = l + h$. Of course, one can obtain an HJB equation in terms of $W$ and $h$. This equation looks as follows
\begin{eqnarray}
V_t (t, W, h) &+& \frac{1}{2}\eta^2h^2V_{hh} (t, W, h) + (rW + (\mu - r) h)V_W (t, W, h)  +  \nonumber \\
  &+&\eta^2h^2 V_{W h} (t, W, h) + \eta^2h^2 V_{W h} (t, W, h) + (\mu - \delta) hV_h (t, W, h) + \nonumber\\
  &+& \max_{\pi} G[\pi] + \max_{c \geq 0} H[c] = 0,\nonumber \\
G[\pi] &=& \frac{1}{2}V_{WW}(t, W, h)\sigma^2 \pi^2 + \left(V_{Wh}(t, W, h) + V_{WW}(t, W, h)\right)\eta\rho\sigma h \pi \nonumber \\
      &+&  +  \pi(\alpha - r)V_W(t, W, h), \nonumber \\
H[c] &=& -cV_W(t, W, h) + \bPhi(t) U(c) \nonumber
\end{eqnarray}
and becomes (\ref{eq:HJB21}) if one changes the variables $\left(W, h\right) \rightarrow \left(l, h\right)$. We will work with the equation (\ref{eq:HJB21}) because it corresponds to a number of works mentioned in introduction and because it is shorter and simpler.}
\end{remark}

\subsection{Viscosity solution of the problem. Comparison Principle}

\begin{definition} \label{defi:viscosity}
	A continuous function $u: \Omega \to \mathbb{R}$ is a viscosity solution of (\ref{eq:HJB21}) if both conditions are satisfied
	\begin{itemize}
		\item
		$u$ is a \emph{viscosity subsolution}, i.e. for any $\phi \in C^2(\bar{\Omega})$ and any local maximum point $z_0 \in \bar{\Omega}$ of $u - \phi$ holds,
		$F(z_0, u(z_0), D\phi(z_0), D^2\phi(z_0)) \leq 0$
		
		\item
		$u$ is a \emph{viscosity supersolution}, i.e. for any $\phi \in C^2(\bar{\Omega})$ and any local minimum point $z_0 \in \bar{\Omega}$ of $u - \phi$ holds,
		$F(z_0, u(z_0), D\phi(z_0), D^2\phi(z_0)) \geq 0$
		
	\end{itemize}
\end{definition}
The fact that the value function for a problem of such kind is a viscosity solution is well known (see e.g. \cite{userguide}) and generally holds if the control and state variables are uniformly bounded. However, this is not the case for the optimal consumption problem and thus a more sophisticated proof is needed.   This area was profoundly studied in \cite{ZaripDuf}, \cite{ZaripDufFlem}, \cite{Zarip1999}. Though our case is similar to the one considered in \cite{ZaripDuf}, the main difference comes from the non-exponential time discounting we are using in the utility functional (\ref{valueFun}). As we mentioned before, this leads to the HJB equation (\ref{eq:HJB21}) being three dimensional. This demands additional work. We will concentrate on the new results and will omit the details of the arguments that work in our problem and could be found in \cite{ZaripDuf}.

\begin{theorem}\label{MT}
There exists a unique viscosity solution of the corresponding HJB equation (\ref{valueFun}) if
\begin{enumerate}
\item $U(c)$ is strictly increasing, concave and twice differentiable in $c$,
\item $\lim_{T \to \infty} \overline{\Phi} (T) E[U(c(T))]=0$, $\overline{\Phi} (T) \sim e^{-\kappa T}$ or faster as $T \to \infty$,
\item $U(c) \leq M(1+c)^{\gamma}$ with $0<\gamma<1$ and $M>0$,
\item $\lim_{c \to 0} U'(c) = + \infty$,  $lim_{c \to + \infty} U'(c) = + \infty$.
\end{enumerate}
\end{theorem}
The proof of this statement is to be done in three steps. At first we need to establish certain properties of the value-function $V(t, l, h)$ that corresponds to our problem. This properties are formulated and proved in Lemma \ref{PVF} that follows. Then we show that the value function with such properties is a viscosity solution of the problem, this is done in Lemma \ref{lm:viscosity}. The uniqueness of this solution follows from the {\em comparison principle} that is actually a very useful tool by itself and is formulated and proved in Theorem \ref{thm:Comparison}. This reasoning is very close to the proof of Theorem 4.1 from \cite{ZaripDuf}.

\begin{lemma}\label{PVF}
Under the conditions $(1) - (4)$ from Theorem \ref{MT} the value function $V(t, l, h)$ (\ref{valueFun}) has the following properties:
\begin{enumerate}
\item[\textup{(i)}] $V(t, l, h)$ is concave and non-decreasing in $l$ and in $h$,
\item[\textup{(ii)}] $V(t, l, h)$ is strictly increasing in $l$,
\item[\textup{(iii)}] $V(t, l, h)$ is strictly decreasing in $t$ starting from some point,
\item[\textup{(iv)}] $0 \leq V(t, l, h) \leq O(|l|^{\gamma} + |h|^{\gamma})$ uniformly in $t$.
\end{enumerate}
\end{lemma}
The proof of the properties $(i) - (ii)$, Lemma \ref{PVF}, can be found in \cite{Zarip1992} and could be applied to our problem with very slight modifications yet we place them here for the consistency of the narrative.
\begin{enumerate}
\item \begin{proof} Let us look on the points $(l_1, h_1)$ and $(l_2, h_2)$ with corresponding $(\pi^{\epsilon}_1, c^{\epsilon}_1)$ and $(\pi^{\epsilon}_2, c^{\epsilon}_2)$ which are $\epsilon$-optimal controls in each of this points respectively or in another words:
$$V(t, l_1, h_1) \leq E \left[ \int^{+\infty}_{t} \bPhi(\tau) U(c^{\epsilon}_1)d\tau \right] + \epsilon,$$
$$V(t, l_2, h_2) \leq E \left[ \int^{+\infty}_{t} \bPhi(\tau) U(c^{\epsilon}_2)d\tau \right] + \epsilon.$$
We choose the point $(\alpha c^{\epsilon}_1 + (1-\alpha) c^{\epsilon}_2)$, where $\alpha \in \mathbb{R}$ and $0 < \alpha < 1$. The policy $(\alpha l_1 + (1-\alpha) l_2, \alpha h_1 + (1-\alpha) h_2)$ is admissible for this point
\begin{eqnarray}
&& V(t, \alpha l_1 + (1-\alpha) l_2, \alpha h_1 + (1-\alpha) h_2) \geqslant \\ \nonumber
&\geqslant& E \left[ \int^{+\infty}_{t} \bPhi(\tau) U(\alpha c^{\epsilon}_1 + (1-\alpha)  c^{\epsilon}_2)d\tau \right]. \label{Ineq1}
\end{eqnarray}
The utility function is concave (see condition $1.$ from Theorem \ref{MT}), so we can write
\begin{eqnarray}
&& E \left[ \int^{+\infty}_{t} \bPhi(\tau) U(\alpha c^{\epsilon}_1 + (1-\alpha)  c^{\epsilon}_2)d\tau \right] \geqslant \\ \nonumber
&\geqslant& \alpha E \left[ \int^{+\infty}_{t} \bPhi(\tau) U( c^{\epsilon}_1) d\tau \right] + (1-\alpha) E \left[ \int^{+\infty}_{t} \bPhi(\tau) U(c^{\epsilon}_2) d\tau \right] \geqslant \\ \nonumber
&\geqslant& \alpha V(t, l_1, h_1) + (1-\alpha) V(t, l_2, h_2) + 2\epsilon.
\end{eqnarray}
Now that we have proved the concavity of $V(t, l, h)$ in $l$ and $h$. We can show that it is not decreasing. Without any loss of generality we can assume that $l_1 \leqslant l_2$ and $h_1 \leqslant h_2$. Note that if $(\pi^{\epsilon}_1, c^{\epsilon}_1)$ is $\epsilon$-optimal for $(l_1, h_1)$ it is admissible for $(l_2, h_2)$ which means that
$$V(t, l_1, h_1) \leqslant V(t, l_2, h_2) + \epsilon,$$
setting $\epsilon \rightarrow 0$ we get that $V(t, l, h)$ is non-decreasing in first two variables. $\bullet$
\end{proof}
\item \begin{proof} To show that $V(t, l, h)$ is strictly increasing in $l$ we can assume the contrary. Let us look at $l_1 < l_2$ such that $V(t, l_1, h) = V(t, l_2, h)$. Since we already know that $V(t, l, h)$ is non-decreasing in $l$ the function $V$ should be constant on the interval $[l_1, l_2]$, moreover, since $V$ is concave in $l$ this interval has to be infinite. This means that there is such $l_0$ that $V(t, l, h) = V(t, l_0, h)$ for any $l \geqslant l_0$. Let $(\pi^{\epsilon}, c^{\epsilon})$ be $\epsilon$-optimal for $(t, l_0, h)$
\begin{equation}\label{constint}
V(t, l_0, h) \leq E \left[ \int^{+\infty}_{t} \bPhi(\tau) U(c^{\epsilon})d\tau \right] + \epsilon.
\end{equation}
We denote $\int^{+\infty}_{t} \bPhi(\tau) d\tau$ as $K(t)$ and look on the inequality
$$l_1 > \max \left(l_0, U^{-1}\left[ 1/K(t) \left( E \left[ \int^{+\infty}_{t} \bPhi(\tau) U(c^{\epsilon}) d\tau \right] + \epsilon \right) \right]/r \right),$$
where $U^{-1}$ denotes an inverse utility function.
The strategy $\pi = 0$ and $c = r l_1$ does not depend on time but is admissible for $(t, l_1, h)$. Indeed, due to the fact that the strategy $(0, rl_1)$ does not depend on time one can write
$$K(t)U(r l_1) = E \left[  \int^{+\infty}_{t} \bPhi(\tau) U(r l_1) d\tau \right] \leqslant V(t, l_1, h).$$
But if we look at $K(t) U(rl_1)$ and use the formula for $l_1$ given above  we get
$$K(t) U(rl_1) > E \left[  \int^{+\infty}_{t} \bPhi(\tau) U(c^{\epsilon})d\tau \right] + \epsilon,$$
which is greater or equal to  $V(T, l_0, h)$ according to the Equation (\ref{constint}). That gives us $V(t, l_0, h) < V(t, l_1, h)$ which is a contradiction keeping in mind that $l_1 > l_0$. So, $V$ is strictly increasing in $l$. $\bullet$
\end{proof}
\item \begin{proof} According to condition $2$ from Theorem \ref{MT} the product of $\bPhi(t)$ and $U(c(t))$ as well as $\bPhi(t)$ itself should be both decreasing for $t > \tau$ starting from a large enough $\tau$. So we choose two moments of time $t_1$ and $t_2$ such that $\tau < t_1 < t_2$, $\Delta t = t_2-t_1$ and look at $V(t_2, l, h)$ then
	\be
		V(t_2, l, h) = \int_{t_2}^\infty \bPhi(t) U(c_t) dt &\stackrel{\tau = t - \Delta t}{=}& \int_{t_1}^\infty \bPhi(\tau + \Delta t) U(c_{\tau + \Delta t}) d\tau \\
	\ee
since $\bPhi(t)$ is decreasing for every $t > t_1$ and the process $c_{\tau + \Delta t}$ for $\tau \geq t_1$ with $L(t_2) = l, H(t_2) = h$ has exactly the same realisations as $c_\tau$ for $\tau \geq t_1$ with $L(t_1) = l, H(t_1) = h$ one can write
	\be
		\int_{t_1}^\infty \bPhi(\tau + \Delta t) U(c_{\tau + \Delta t}) d\tau & <& \int_{t_1}^\infty \bPhi(\tau) U(c_{\tau + \Delta t}) d\tau \leqslant V( t_1, l, h).
	\ee
So for any $t_1$ and $t_2$ such that $\tau < t_1 < t_2$ we get $V(t_1, l, h) > V(t_2, l, h)$. $\bullet$
\end{proof}
\item \begin{proof} The strategy for proving the upper bound is based on the ideas of Huang, Pag\`es \cite{HuangPages} and Duffie, Zariphopolou \cite{ZaripDuf}. Instead of the original problem with the non-traded income generated by $H_t, H_0 = h$ one can consider a fiction consumption-investment problem with a special asset on the market, such that has a sufficient initial endowment (meaning that one can generate exactly the same income flow as $H_t$ would by investing in the market). Suppose the synthetic asset follows geometrical Brownian motion
\begin{equation}
	dS'_t = \alpha'S'_t + \sigma' S'_tdW_t, \quad t \geq 0 ~~ S'_0 = s', \quad s' > 0,
\end{equation}
with constants $\alpha'$ and $\sigma'$ to be defined later. Next, the \emph{initial wealth equivalent} of the stochastic income is defined by
\begin{equation*}
	f(h) = \delta E_h\left[ \int^{\infty}_0 e^{-\kappa t}\xi_t H_t dt \right],
\end{equation*}
where
$$
	\xi_t = \exp\left(-\frac{1}{2}(\theta^2_1 + \theta^2_2) + \theta_1W^{1}_t + \theta_2W^{2}_t\right),
$$
$\theta_1 = (\alpha-r)/\sigma$ and $\theta_2 = (\alpha' - r)/\sigma'$.

 It turns out that with the properly chosen $\alpha'$ and $\sigma'$ we achieve that $f(h) < C_1 h$. Moreover, the stochastic income rate $H_t$ can be replicated by a self-financing strategy on the complete market $(B_t, S_t, S'_t)$ with the additional initial endowment $f(h)$. This fact is well known from the martingale-based studies of the consumption-investment problem, primarily carried out in \cite{HuangPages} and \cite{KaratzasShreve}.

To finish the proof, we notice that since the stochastic income can be replicated, any admissible strategy for the original problem with initial conditions $(l, h)$ is dominated by a strategy on the synthetic market with initial endowment $l + f(h) < l + C_1h$. On the other hand, we have the growth conditions for $\bPhi(t)$ and $U(c)$. So, the maximal utility is bounded from above by the solution of the classic investment-consumption problem with initial wealth $l + C_1h$, HARA utility and exponential discounting. Due to Merton we have a closed form solution for this case. Putting everything together, we obtain the desired bound (all the further details can be found in \cite{ZaripDuf}). $\bullet$
\end{proof}
\end{enumerate}

 Now we can prove the existence of the viscosity solution of the problem  (\ref{eq:HJB21}).
\begin{lemma} \label{lm:viscosity}
	Under the conditions of Lemma \ref{PVF} the function $V(t, l, h)$ is a viscosity solution of (\ref{eq:HJB21}) on the domain $D = (0, \infty)
	\times (0, \infty) \times (0, \infty)$.
\end{lemma}
\begin{proof}
	We again use the reasoning from the proof of Theorem 4.1 in \cite{ZaripDuf} but modify it for our case. To show that $V$ is a viscosity solution one need to show that it is a viscosity supersolution and a viscosity subsolution of the problem.\\
	Let us show at first that $V(t, l, h)$ is a viscosity supersolution for (\ref{eq:HJB21}). Let us look at $\phi \in C^2(D)$ and assume that $(t_0, l_0, h_0) \in D$ is a point where a minimum of $V-\phi$ is achieved.
	We can assume that $V(t_0, l_0, h_0) = \phi(t_0, l_0, h_0)$ and $V > \phi$ in $D$ without any loss of generality.
	To show that $V$ is a supersolution we need to check that $\mathcal{F}[\phi](t_0, l_0, h_0, \pi, c) \leq 0$, where
	\ben
	&& \mathcal{F}[\phi](t_0, l_0, h_0, \pi, c) =	\phi_t (t_0, l_0, h_0) + \frac{1}{2}\eta^2h_0^2\phi_{hh} (t_0, l_0, h_0) + \nonumber \\
																 	 &+&  (rl_0 + \delta h_0)\phi_l (t_0, l_0, h_0) + (\mu - \delta) h_0\phi_h (t_0, l_0, h_0) +  \nonumber \\
																	 &+& \max_{\pi} G[t_0, l_0, h_0, \pi] + \max_{c} H[t_0, l_0, h_0, c], \nonumber  \label{eq:subbound}
	\een
with $G[\pi]$ and $H[c]$ defined in (\ref{eq:HJB21}).

	 We consider a locally constant strategy $(\pi_0, c_0)$ for the period of time $\theta$ tending to zero. One can take $\theta = \min\{1/n, \tau\}$ where $\tau = \inf\{t \geq t_0: W_t = 0\}$ to ensure feasibility of this strategy. Since this strategy is suboptimal we can write (using the dynamic programming principle, \cite{FlemingSoner})
	\begin{eqnarray}
		V(t_0, l_0, h_0) &\geq& E\lb[ \int_{t_0}^{t_0 + \theta} \bPhi(t)U(c_0)dt + V(L_\theta, H_\theta, \theta) \rb]  \label{eq:DPP1} \\
										 &\geq& E\lb[\int_{t_0}^{t_0 + \theta} \bPhi(t)U(c_0)dt + \phi(L_\theta, H_\theta, \theta) \rb]. \nonumber
	\end{eqnarray}
	
	On the other hand, applying It\^o calculus to the smooth function $\phi$ we can expand
	$$
		E[\phi(\theta, L_\theta, H_\theta)] = \phi(t_0, l_0, h_0) + E\lb[\int_{t_0}^{t_0 + \theta} D\phi(s, L_s, H_s) ds\rb].
	$$
	Substituting into (\ref{eq:DPP1}) and using standard estimates to approximate the terms with $\phi(s, l_s, h_s)$, $\phi_l(s, l_s, h_s)$, $\phi_h(s, l_s, h_s)$, etc. via $\phi(t_0, l_0, h_0) + O(s)$, $\phi_l(t_0, l_0, h_0) + O(s)$, $\phi_h(t_0, l_0, h_0) + O(s)$ respectively, we obtain the bound
	$$
		E\lb[ \int_{t_0}^{t_0 + \theta} \mathcal{F}[\phi](t_0, l_0, h_0, c_0, \pi_0) \rb] + E\lb[\int_{t_0}^{t_0 + \theta} h(s) ds\rb] \leq 0,
	$$ 	
	with $h(s) = O(s)$. Dividing by $E[t_0 + \theta]$ and taking the limit $n \to \infty$ (so $\theta \to 0$ and $E\lb[\int_{t_0}^{t_0 + \theta} h(s) ds\rb] \to 0 $) we get (\ref{eq:subbound}) as $(\pi_0, c_0)$ can be arbitrary admissible pair.

	The second part of the proof is to show that $V(t, l, h)$ is a subsolution as well. However, the proof of the second part of Theorem 4.1 in \cite{ZaripDuf} can be applied verbatim here so we omit further details. $\bullet$
	
\end{proof}

The third result that is needed to finalize the proof of Theorem \ref{MT} is a {\em comparison principle} formulated below as Theorem \ref{thm:Comparison}. Results of this type are well-known in general for bounded controls, but due to the unboundness of the controls, classical proofs require adaptations for our case.
\begin{theorem} \label{thm:Comparison} \textbf{({\em Comparison Principle})}
Let $u(t, l, h)$ be an upper-semicontinuous concave viscosity subsolution of (\ref{eq:HJB21})
on $D$ and $V(t, l, h)$ is a supersolution of (\ref{eq:HJB21}) on $D$ which is bounded from below,
uniformly continuous on $D$, and locally Lipschitz in $D$, such that $u(t, l, h) \to 0$, $V(t, l, h) \to 0$ as $t \to \infty$ and $|u(t, l, h)| + |V(t, l, h)| \leq O(|l|^\gamma + |h|^\gamma)$ for large $l, h$, uniformly in $t$. Then $u \leq v$ on $\bD$.
\end{theorem}
\begin{proof}
	Let us introduce $x:= (l, h),~ x \in \mathbb{R}^+ \times \mathbb{R}^+$ to make formulae shorter. Assume for contradiction that
	$$
		\sup_{(t, x) \in \bD} [u(t, x) - v(t, x)] > 0.
	$$
	Let $T_n \to \infty$ be an increasing sequence of time moments, $m > 0$ be a parameter and
	$$
		\Psi^{m, n}(t, x) = u(t, x) - v(t, x) - m(T_n - t).
	$$
 Since $u, v \to 0$ as $t \to \infty$, for sufficiently large $n$ and sufficiently small $m$ the maximum of $\Psi^{m, n}$ must occur in an internal point of $D$. So let us assume that $\bar{m} > 0$ and $T_n$ are such that $\sup_{(x, t) \in \bD} \Psi^{\bar{m}, n}(x, t)$ occurs in some point $(t_0, x_0)$ with $t_0 < T_n$.
	Let us define two functions
	$$
		\tilde{u}(t, x) = u(t, x) - \bar{m}(T_n-t)
	$$
	and
	$$
		\phi(t, x, y) = \lb| \frac{y - x}{\xi} - 4\varpi\rb|^4 + \theta(l_x + h_x)^\lambda + \bar{m}(T_n - t)
	$$
	where $x = (l_x, h_x)$, $y = (l_y, h_y)$ and $\lambda \in (\gamma, 1)$,
	$\theta, \xi > 0$, $\varpi \in \mathbb{R}_+^2$ being parameters to be varied later.
	Finally, we look at the point $(\bar{x}, \bar{y}, \bar{t})$ where the following function achieves a maximum
	$$
		\psi(t, x, y) = \tilde{u}(t, x) - v(t, y) - \phi(t, x, y).
	$$
	Since $\bar{t}$ is an interior point we can write
	\begin{equation}
		2\bar{m} = u_t(\bar{t}, \bar{x}) - v_t(\bar{t}, \bar{y}). \label{eq:uvbound}
	\end{equation}

	On the other hand, one can bound $u_t( \bar{t}, \bar{x}, \bar{t}) - v_t(\bar{y})$ merely by $\phi$ and its derivatives which can be written down explicitly.
It appears then, that as $\theta, \xi, \|\varpi\| \to 0$ the distance $\|\bar{x} - \bar{y}\|$ tends to zero and both $(\bar{t}, \bar{x})$,
$(\bar{t}, \bar{y})$ are close to $(t_0, x_0)$, so in the limit in terms of $\|\bar{x} - \bar{y}\| \to 0$ (\ref{eq:uvbound}) leads to $\bar{m} \leq 0$ and we get a contradiction. Again, further technical details are omitted and can be found in the proofs of Theorem 4.2 in \cite{ZaripDuf} and Theorem 3.2 in \cite{Zarip1999}. $\bullet$
\end{proof}

\subsection{Homotheticity reduction for utility functions of the HARA type}

Though the HJB equation (\ref{eq:HJB21}) generally fails to have a reduction with respect to the time variable, it is possible to reduce the dimension if the utility function is of the HARA type. In this paper we will develop further analysis just for the logarithmic utility function. This has several reasons. First of all, as we have already mentioned, the logarithmic case allows one to consider time distributions with subexponential tails, while enjoying the homotheticity reduction available for utility functions of the general HARA type. Secondly, the logarithmic case could in some sense be regarded as a limiting of the HARA case with $\gamma$ tending to zero. This allows to translate all the obtained results to the general power case of HARA utility with only straightforward modifications.

Rewriting the HJB equation (\ref{eq:HJB21}) for the logarithmic utility function $U(c(t))=\log {c(t)}$ we get
\begin{eqnarray}
  V_t(t, l, h) +  \frac{1}{2}\eta^2h^2V_{hh} (t, l, h) &+& (rl + \delta h)V_l (t, l, h) + (\mu - \delta) hV_h (t, l, h)  \nonumber \\
  &+& \max_{\pi} G[\pi] + \max_{c \geq 0} H[c] = 0 \label{eq:HJBLogInfty21}  \\
G[\pi] &=& \frac{1}{2}V_{ll}(t, l, h)\pi^2\sigma^2 + V_{lh}(t, l, h)\eta\rho\pi\sigma h   \\
       &+& \pi(\alpha - r)V_l(t, l, h), \nonumber \label{eq:GLogInfty}   \\
H[c] &=& -c V_l (t, l, h) + \bPhi(t)\log(c). \label{eq:HLogInfty}
\end{eqnarray}
Using the homotheticity of the logarithm and homogeneity of the differential operator applied to the value function in (\ref{eq:HJBLogInfty21}) we rewrite $V(t, l, h)$ in the following way
\begin{eqnarray}
V(t, l, h) = W(t, z) - \Psi_1(t)\log h + \Psi_2(t),\label{substitution_log}
\end{eqnarray}
having $z = l/h$ and $\Psi_1(t) =  \int_t^{\infty} \bPhi(s)ds$ and $\Psi_2(t)$ to be chosen later.

\begin{remark} {\em The form of the substitution can be defined via Lie group analysis of the given equation. This analysis for logarithmic and general HARA-type utility is done in \cite{boya}. }
\end{remark}

The Hamiltonian terms $\max_{\pi} G[\pi]$ and $\max_{c} H[c]$ in (\ref{eq:HJBLogInfty21}) now become
\begin{eqnarray}
	\max_{\pi} G[\pi]						&=& \max_{\pi' = \pi/h} \lb[ \frac{1}{2} W_{zz}\sigma^2 \pi'^2 + \pi'\lb(-\eta\rho\sigma(W_z+zW_{zz}) + (\alpha-r)W_z\rb) \rb], \label{eq:maxPi} \\
	\max_{c} H[c]							&=&	\max_{c' = c/h} [-c'W_z + \bPhi(t)\log(c')] + \bPhi(t)\log(h) \label{eq:maxH},					 	
\end{eqnarray}
and the optimal policies after formal maximization are
\ben
	\pi_{\star}(l, h) &=& h \sigma^{-2} \lb(\eta\rho\sigma z - ((\alpha -r) - \eta\rho\sigma)\frac{W_z}{W_{zz}}  \rb), \label{eq:optimalPiLoggen} \\
	c_{\star}(l, h) &=&  h\frac{\bPhi(t)}{W_z}, \label{eq:optimalCLoggen}
\een

We rewrite (\ref{eq:HJBLogInfty21}) using formulae (\ref{eq:maxPi})  and (\ref{eq:maxH})
\begin{eqnarray}
  W_t &+& \Psi_2'(t) + \lb(-\frac{\eta^2}{2} +(\mu-\delta)\rb)\Psi_1(t) + \frac{\eta^2}{2}z^2W_{zz} + (\eta^2 + r - (\mu-\delta))z W_z + \delta W_z  \nonumber\\
  		&+& \max_{\pi'}\lb[ \frac{1}{2} W_{zz}\sigma^2 \pi'^2 + \pi'\lb(-\eta\rho\sigma(W_z+zW_{zz}) + (\alpha-r)W_z\rb) \rb] \nonumber \\
  		&+& \max_{c' \geq 0}\left[-c'W_z + \bPhi(t)\log (c') \right] =  0. \label{eq:HJB222}
\end{eqnarray}

We provide the formal maximization of $H[\pi]$ and $G[c]$ and obtain
\be
	\max_{\pi} H[\pi] &=& -\frac{1}{2}\lb( (\eta \rho - (\alpha - r)/\sigma)^2\frac{W_z^2}{W_{zz}} + 2 \eta \rho( \eta \rho-(\alpha - r)/\sigma)zW_z +  \eta \rho^2z^2W_{zz}\rb) \\
	\max_{c} G[c] 	&=& \bPhi(t)\left(\log \bPhi(t) - 1 \right) - \bPhi(t) \log W_z.
\ee
so (\ref{eq:HJB222}) becomes
\begin{eqnarray}
  W_t &+& \Psi_2'(t) + \lb(-\frac{\eta^2}{2} +(\mu-\delta)\rb)\Psi_1(t) + \bPhi(t)(\log \bPhi(t) - 1)v + \nonumber \\
  		&+& d_2 z^2W_{zz} - \frac{d_1^2}{2}\frac{(W_z)^2}{W_{zz}} + d_3 z W_z + \delta W_z - \bPhi(t)\log W_z = 0, \label{eq:subti}
\end{eqnarray}
where
\begin{eqnarray} \label{constd}
	d_1 &=& \frac{\alpha - r  - \eta\rho\sigma}{\sigma^2},  \quad
d_2 =  \frac{1}{2} \eta^2(1-\rho^2),  \\
d_3 &=&  2 d_2 + \frac{\rho \eta}{\sigma}(\alpha - r) + r - (\mu - \delta). \nonumber
\end{eqnarray}

Now by choosing $\Psi_2(t)$ as a solution of the equation
 \begin{eqnarray} \label{eqnpsi2}
 \Psi_2'(t) + \lb(-\frac{\eta^2}{2} +(\mu-\delta)\rb)\Psi_1(t) + \bPhi(t)(\log \bPhi(t) - 1)=0 ,\\
 \Psi_2(t)\to 0, t\to \infty, \nonumber
 \end{eqnarray}
we can cancel out the terms dependent only on $t$ in the equation (\ref{eq:subti}). We arrive at
\begin{equation} \label{eq:HJBfinal}
 W_t - \frac{d_1^2}{2}\frac{(W_z)^2}{W_{zz}}  + d_2 z^2W_{zz} + d_3 z W_z + \delta W_z - \bPhi(t)\log W_z = 0.
\end{equation}

%In the similar way we use a substitution for the general HARA function and obtain the reduced equation for the general HARA type utility function
%\begin{eqnarray} \label{general_utility_hara}
%  &&W_t  - \frac{d_1^2}{2}\frac{(W_z)^2}{W_{zz}} + d_2 z^2W_{zz} + d_3 z W_z  + \delta W_z - \nonumber\\
%  &-& \bPhi(t) U'(c(t)) W_z = 0 \nonumber.
%\end{eqnarray}

\subsection{Bounds for the value function} \label{bounds}

The main tool we are going to use to obtain the bounds is the comparison principle given by Theorem \ref{thm:Comparison}. Since (\ref{eq:HJBfinal}) is a two-dimensional PDE and by itself is not a HJB equation, we argue as follows. Any formal sub- or super- solution of (\ref{eq:HJBfinal}) can be transformed to a sub- or super- solution of (\ref{eq:HJBLogInfty21}) with a substitution described by (\ref{substitution_log}). On the other hand, for the HJB equation (\ref{eq:HJBLogInfty21}) Theorem \ref{MT} and Theorem \ref{thm:Comparison} hold and we can obtain a lower and upper bound.

To shorten the notation we will use
\begin{eqnarray}
 && F(z, W, D W, D^2W) = \label{eq:HJBfinalMinus} \\
 && -W_t + \frac{d_1^2}{2}\frac{(W_z)^2}{W_{zz}} - d_2 z^2W_{zz}  - d_3 z W_z - \delta W_z + \bPhi(t)\log W_z.\nonumber
\end{eqnarray}
Note, in order to comply with the Definition \ref{defi:viscosity} we have to take the equation (\ref{eq:HJBfinal}) with the minus sign.

Determining an upper bound demands specific information on the cumulative distribution function {\em $\Phi (t)$} of the liquidation time. In the next Section this issue is addressed specifically for the cases of exponentially distributed liquidation time $\tau$ and Weibull distributed liquidation time. This two distributions seem to be the most practically applicable to the asset with low liquidity.

A lower bound, however, could be found without any specific information on $\Phi(t)$. Let us look on an optimal consumption problem without random income. This is a classical two dimensional Merton's problem for which we can write the HJB equation on the value function $u(t, z)$0. This problem corresponds to (\ref{eq:HJBLogInfty21}) but without any terms, containing the derivatives with respect to $h$ and with a notation $V \rightarrow u$, $l \rightarrow z$
\begin{eqnarray}
u_t &+& rlu_z + \max_{\pi} G[\pi] + \max_{c \geq 0} H[c] = 0, \label{eq:HJB}\\
G[\pi] &=& \frac{1}{2}u_{zz}(t, z)\pi^2\sigma^2 + \pi(\alpha - r)u_z(t, z), \label{eq:Gmax} \\
H[c] &=& -cu_z (t, z) + \bPhi(t)\log(c). \label{eq:Hmax}
\end{eqnarray}

After the formal maximization, one gets
$$
u_t + rlu_z - \frac{1}{2}\lb(\frac{\alpha-r}{\sigma}\rb)^2 \frac{u_z^2}{u_{zz}} + \bPhi (t)\left(\log \bPhi(t) - \bPhi(t)\right) - \bPhi(t) \log u_z = 0.
$$

We look for a solution in the form $u(t, z) = \Psi_1(t) \log z + \Theta_1(t)$, where again $\Psi_1(t) = \int_t^\infty \bPhi(s) ds$ and $\Theta_1(t)$ is a solution of
\begin{eqnarray}\label{eqnpsi3}
\Theta'_1 + \Psi_1 \lb(r + \frac{1}{2}\frac{(\alpha-r)^2}{\sigma^2}\rb) - \bPhi(\bPhi - \log \bPhi + \log \Psi_1) = 0.
\end{eqnarray}

One can easily check that such $u$ tends to zero uniformly as $t \to \infty$ and since the solution of (\ref{eq:HJB}) is a lower bound for our three-dimensional problem we obtain the following inequality for the lower bound
$$
\Psi_1(t) \log z + \Theta_1(t)  \leq W(z, t) = V(t, l, h) - \Psi_1 \log h + \Psi_2(t),
$$
or
$$
\Psi_1(t) \log l + \Theta_1(t) - \Psi_2(t)  \leq V(t, l, h).
$$
In the next Sections we consider specific liquidation time distributions. First we take the most simple one - an exponential distribution where a lot more can be said. In particular, we get asymptotically tight bounds for the value function and derivatives, which lead to asymptotic formulae for the optimal policies.
Not surprisingly, in the limit case when the random income vanishes the value function and optimal policies coincide with the classical Merton solution for the logarithmic case.

Another somewhat more complicated case is the Weibull distribution, where the bounds have no elementary representation, but their asymptotic can be derived using incomplete gamma functions.

\section{The case of exponential distributed liquidation time and logarithmic utility function} \label{expsec}
Now we examine the optimal consumption problem introduced before in the case of the logarithmic utility. Despite that we know from the more general theorems from \cite{ZaripDuf} that the optimal strategy does exist and the value function is the viscosity solution of the HJB equation, it is desirable to have the optimal policy in the feedback form (\ref{eq:optimalPiLoggen}) and (\ref{eq:optimalCLoggen}). In a general situation the feedback optimal policy is hard to establish since the value function is not a priori smooth. On the other hand, smoothness of the value function simplifies the problem so it becomes amenable to standard verification theorems of optimization theory, see e.g. \cite{FlemingSoner}. Here we prove that in the case at hand the value function is twice differentiable. As far as we know this fact was not explicitly addressed before, though the structure of our proof is similar to the paper \cite{ZaripDufFlem} where the smoothness was proved for the HARA utility case. Since the case without stochastic income is known to have a closed form solution and was derived by Merton \cite{Merton}, it is plausible to consider it as a zero-term approximation. Keeping that in mind, we will rigorously prove that value function tends to the Merton closed form solution in the limit of vanishing random income.

\subsection{Reduction of the HJB equation}

Recall the definition of the value function
\begin{equation} \label{valueFunLogExp}
	V(t, l, h) = \max_{(\pi, c) } E \left[ \int_t^\infty e^{-\kappa t}\log(c) dt | L(t) = l, H(t) = h \right].
\end{equation}
At first let us note that in the exponential liquidation time distribution case the problem is homogenous in time. We introduce $\tilde{V}(l, h)$
\be
	\tilde{V}(l, h) &=& \max_{(\pi, c) } E \left[ \int_t^\infty e^{-\kappa (s-t)}\log(c) ds \right] \\
										 &=& \max_{(\pi, c) } E \left[ \int_0^\infty e^{-\kappa v}\log(c) dv \right],
\ee
which is independent on time. Substituting
$$
V(t, l, h) = e^{-\kappa t}\tilde{V}(l, h)
$$
into the HJB equation (\ref{eq:HJB21}) we arrive at a time-independent PDE on $\tilde{V}(l, h)$. With a slight abuse of notation, hereafter we will use the same letter $V$ for $\tilde{V}$. The reduced equation takes the form

\begin{eqnarray}
 \frac{1}{2}\eta^2h^2V_{hh} (l, h) &+& (rl + \delta h)V_l (l, h) + (\mu - \delta) hV_h (l, h)  \nonumber \\
  &+& \max_{\pi} G[\pi] + \max_{c \geq 0} H[c] = \kappa V(l, h), \label{eq:HJBLogInfty}\\
G[\pi] &=& \frac{1}{2}V_{ll}(l, h)\pi^2\sigma^2 + V_{lh}(l, h)\eta\rho\pi\sigma h \\
       &+& \pi(\alpha - r)V_l(l, h), , \label{eq:GLogInfty2Dim}\  \\
H[c] &=& -cV_l (l, h) + \log(c). \nonumber
\end{eqnarray}

Now using substitution (\ref{substitution_log}) with $\Psi_1 =  \frac{1}{\kappa}$ and  $\Psi_2 = \frac{1}{\kappa^2}\left(\mu - \delta - \frac{\eta^2}{2}\right)$  we can argue exactly as in the general case and represent $V(l, h)$ in the form
\begin{equation}\label{eq:ValueFuncLog}
	V(l, h) = v(z) + \frac{\log h}{\kappa} +\frac{1}{\kappa^2}\left(\mu - \delta - \frac{\eta^2}{2}\right), ~~z = l/h,
\end{equation}
so $v(z)$ satisfies the equation
\begin{eqnarray}
  \frac{\eta^2}{2}z^2v'' &+& \max_{\pi}\left[\frac{1}{2}\pi^2\sigma^2v' - \pi\left((v' + zv'')\eta\rho\sigma + (\alpha-r)v' \right)\right] \nonumber \\ &+& \max_{c \geq -\delta}\left[-cv_z + \log (c + \delta) \right] =  \kappa v, \label{eq:HJBLogMax}
\end{eqnarray}
where $v' = v_z$ and the dimension of the problem is reduced to one. It is important to note that such reduction was possible due to the exponential decay, the homotethicity of the logarithmic function and the linearity of the control equations, which make the reduction (\ref{eq:ValueFuncLog}) sound.

Assuming that $v$ is smooth and strictly concave, we perform a formal maximization of the quadratic part (\ref{eq:GLogInfty2Dim}) which leads to
\begin{equation}
  \kappa v v''= - \frac{d_1^2}{2}(v')^2 + d_2 z^2(v'')^2
           + d_3 z v' v'' - v'' \left[1 + \log (v') \right], \label{eq:logRed}
\end{equation}
where again $d_1, d_2$ and $d_3$ are defined in (\ref{constd}).

Coming back to the original variables we obtain the optimal policies in the form
\begin{eqnarray}
	c_{\star}(l, h) &=& \arg \max_{c \geq 0} (-cV_l + \log c) =  \frac{h}{v'(l/h)}, \label{eq:optimalCLog}\\
	\pi_{\star}(l, h) &=& \arg \max_{\pi} \left( \frac{1}{2}\pi^2V_{ll}\sigma^2 + \pi\left(V_{lh}\eta\rho\sigma h + (\alpha - r)V_l \right) \right) \nonumber \\
	            &=&- \frac{\eta \rho}{\sigma}l - h\frac{d_1}{\sigma}\frac{v'(l/h)}{v''(l/h)}. \label{eq:optimalHLog}
\end{eqnarray}

Summing up, we announce the main result of this Section.
\begin{theorem}
\label{thm:main_exp}
	Suppose $r - (\mu - \delta) > 0$ and $d_1 \neq 0$.
	\begin{itemize}
		\item[\textup{(i)}]
	  	There is the unique $C^2(0, +\infty)$ solution $v(z)$ of (\ref{eq:logRed}) in a class of concave functions.
		\item[\textup{(ii)}]
		  For $l, h > 0$ the value function is given by (\ref{eq:ValueFuncLog}). For $h = 0, l > 0$ the value function $V(l,0)$ coincides with the classical Merton solution
		   \begin{equation}
		  V(l, 0) = \frac{1}{\kappa^2}\left[r + \frac{1}{2}\frac{(\alpha - r)^2}{\sigma^2} - \kappa \right] + \frac{\log(\kappa l)}{\kappa}.\label{eq:MertonSol}
\end{equation}
		
		\item[\textup{(iii)}]
	    If the ratio between the stochastic income and the total wealth tends to zero, the policies $(\pi^{\star}, c^{\star})$ given by (\ref{eq:optimalCLog}), (\ref{eq:optimalHLog}) tend to the classical Merton's policies
	    \begin{eqnarray}
	    c_{\star}(l, 0) \sim \kappa l, \label{eq:optimalCMerton} ~  \pi_{\star}(l, 0) \sim - \frac{(\alpha - r)l }{\sigma^2} \frac{V_l^2}{V_{ll}}. \label{eq:optimalHMerton}
	\end{eqnarray}
		
		\item[\textup{(iv)}]
		  Policies (\ref{eq:optimalCLog}) and (\ref{eq:optimalHLog}) are optimal.
	
	\end{itemize}
\end{theorem}
We have shown that the solution exists and tends to Merton case when $h = 0$. In the next step we will show the smoothness of the solution.

\subsection{The dual optimization problem and smoothness of the viscosity solution}
In this Section we introduce the dual optimization problem with a synthetic asset such that the optimization equation formally coincides with (\ref{eq:logRed}).  The regularity of the dual problem proves the regularity of the original one due to the uniqueness of the viscosity solution.

    Let us consider the investment-consumption problem with the wealth process $Z_t$ defined by
\begin{eqnarray}
	Z_t &=& (d_3 Z_t + d_1 \sigma \pi_t - c_t) dt + \sigma \pi_t dW^1_t + \eta Z_t \sqrt{1 - \rho^2} dW^2_t, \label{eq:dualProcess} \\
	Z_0 &=& z \geq 0, \nonumber
\end{eqnarray}
where $d_1$ and $d_3$ are defined in (\ref{constd}). We define the set of admissible controls $\hat{\mathcal{A}}(z)$ as the set of  pairs $(\pi, c)$ such that
\begin{enumerate}
	\item
	There exists an a.s. positive solution $Z_t$ of the stochastic differential equation (\ref{eq:dualProcess}).
	\item
	$c_t \geq -\delta$.
	\item
	$c$ and $\pi$  satisfy the integrability conditions (\ref{integrability}).
\end{enumerate}

The investor wants to maximize the average utility given by
$$
	\hat{\mathcal{U}}(c) = E \left[\int^{\infty}_0{e^{-\kappa \tau}\log( \delta + c(\tau))}d\tau \right] \label{eq:dualUtility}
$$
and the value function $w$ is defined as
$$
	w(z) = \sup_{(\pi, c) \in \hat{\mathcal{A}}(z)} \hat{\mathcal{U}}(c).
$$
The associated HJB equation is
\begin{equation}
	\kappa w = d_2 z^2w'' + \max_{\pi}\left[\frac{1}{2}\sigma^2\pi^2 w'' + d_1 \sigma \pi w'\right]  + d_3 zw' + \max_{c \geq -\delta}\left[-cw' + \log (c + \delta) \right], \label{eq:dualHJB}
\end{equation}
Next, keeping in mind $w' > 0, w'' < 0$, we can rewrite (\ref{eq:dualHJB}) as
\begin{equation}
	 -\frac{d_1^2}{2} \frac{(w')^2}{w''} + d_2 z^2w'' + d_3 zw' +  \delta w' - 1 - \log w' - \kappa w = 0. \label{eq:dualHJBFull}
\end{equation}

Now, it is easy to see that (\ref{eq:dualHJB}) reduces to (\ref{eq:logRed}) assuming that $w$ is smooth. Thus, if we prove that $w$ is smooth and concave, we will get the desired result for $v$ as well. The possibility to switch back and forth from $V$ to $v$ and $w$ is guaranteed by the existence and uniqueness of the viscosity solutions given by Theorem \ref{MT}. On the other hand, if a function is the value function for the corresponding optimization problem, and the HJB equations formally coincide, the value functions must coincide as well due to uniqueness. Therefore, it is sufficient to prove that $w$ is smooth.

From the previous Section we already know that if $D = (0, \infty)$ and $\bD = [0, \infty]$ the following theorem hold.

\begin{theorem}
 \begin{itemize}
	The function $w$ is the unique viscosity solution of (\ref{eq:dualHJB}) in $D$.
	The value function $V(l, h)$ is the unique viscosity solution of (\ref{eq:ValueFuncLog}) in $D \times D$.
	\end{itemize}
\end{theorem}
Let us now prove the smoothness of the solution and of its' first derivative.

\begin{theorem} \label{thm:smoothness}
	The function $w$ is the unique concave $C^2(D)$ solution of (\ref{eq:dualHJB}).
\end{theorem}

To start with the proof of the theorem we need some explicit bounds for $w$.

\begin{lemma} \label{lm:boundness}
	The following bounds hold for $w(z)$
	\begin{equation}
     C_1\log(z+C_2) < w(z) < (z+C_3)^\gamma, \quad z \in \Omega \label{eq:wBounds}
  \end{equation}
	for some constants $C_1, C_2, C_3 > 0$ and $0 < \gamma < 1$.
\end{lemma}

\begin{proof}
The function
$$
	W^{-}(z) = C_1\log(z + C_2), \quad z \in \Omega
$$
is a subsolution for (\ref{eq:dualHJBFull}) as the coefficient of the leading logarithmic term is negative provided $C_1, C_2 > 0$ are appropriately chosen.
On the other hand, the function
$$
	W^{+}(z) = (z+C_3)^\gamma, \quad z \in \Omega
$$
is a supersolution provided $0 < \gamma < 1$ is sufficiently close to $1$. Indeed, the leading term is $z^\gamma$ with the coefficient $-(d_1^2(w')^2)/(2w'')$, which in turn grows as $-\gamma/(\gamma-1)$ and becomes arbitrarily large as $\gamma$ tends to $1$.

Thus, the desired bound (\ref{eq:wBounds}) is a consequence of the comparison principle formulated in Theorem \ref{thm:Comparison}. $\bullet$
\end{proof}
Now we can prove Theorem \ref{thm:smoothness}.
\begin{proof} It is known that uniformly elliptic equations enjoy regularity, but as before the main obstacle is the lack of uniform bounds. Our proof will closely follow
the approach used in \cite{ZaripDufFlem}, the original problem is approximated  by a convergent family of optimization problems such that the approximating equations is uniformly elliptic and thus smooth. Then the smoothness follows from the stability of viscosity solutions and uniqueness.
	
	\textbf{Step 1.} Consider the value function
	$$
	w_L(z) = \sup_{(\pi, c) \in \hat{\mathcal{A}}(z)} \hat{\mathcal{U}}(c).
  $$
	for the problem with the additional strategy constraint $-L \leq \pi_t \leq L$ for almost every $t$. Arguing as in Section \ref{problem} we conclude that that $w^L$ is an increasing continuous function, which is the unique viscosity solution to
  \begin{eqnarray}
	  \kappa w_L &=& d_2 z^2w_L'' + \max_{-L \leq \pi \leq L}\left[\frac{1}{2}\sigma^2\pi^2 w_L'' + d_1 \pi w_L'\right] \label{eq:dualHJBL} \\
	        &+& d_3 zw_L' + \max_{c \geq -\delta}\left[-cw_L' + \log (c + \delta) \right]. \nonumber
   \end{eqnarray}
 Moreover, the bounds of Lemma ~\ref{lm:boundness} hold so
 \begin{equation}
     C_1\log(z+C_2) < w^L(z) < (z+C_3)^\gamma. \nonumber
 \end{equation}
 Thus, there exists a concave function $\hat{w}$ such that $w_L \to \hat{w}, L \to \infty$ locally uniformly. Then due to the stability property and
 uniqueness of the viscosity solution the function $\hat{w}$ is a viscosity solution of (\ref{eq:dualHJB}) and thus coincides with $w$. Therefore
 $w_L \to w, L \to \infty$ locally uniformly.

 \textbf{Step 2.} We claim that $w_L$ is a smooth function on an arbitrary interval $[z_1, z_2]$ such that $z_1 > 0$. Due to concavity we may assume that derivatives
 $w_L'(z_1), w_L'(z_2)$ exist.
 On the one hand the function $w_L$ is the unique solution of the boundary problem
  \begin{eqnarray}
	  \kappa u &=& d_2 z^2u'' + \max_{-L \leq \pi \leq L}\left[\frac{1}{2}\sigma^2\pi^2 u'' + d_1\sigma \pi u'\right]  \label{eq:elliptic} \\
	        &+& d_3 zu' + \max_{c \geq -\delta}\left[-cu' + \log (c + \delta) \right], \nonumber  \\
	        u(z_1) &=& w_L(z_1), \quad u(z_2) = w_L(z_2), \quad z \in [z_1, z_2]. \nonumber
  \end{eqnarray}
 On the other hand, according to the general theory of fully nonlinear elliptic equations of second order of Bellman type in a compact region, (see {\em Krylov} \cite{krylov}), (\ref{eq:dualHJBL}) has a unique $C^2$ solution in $[z_1, z_2]$ that coincides with $w_L$  and $w_L$ is smooth on $[z_1, z_2]$.

 \textbf{Step 3.}  We show that the constraint $-L \leq \pi_t \leq L$ is superfluous for sufficiently large $L$ and can be eliminated. First it is clear that due to concavity and monotonicity of $w_L$, the condition $-L \leq \pi_t \leq L$ in (\ref{eq:dualHJBL}) can be substituted with $\pi_t \leq L$.
Now we prove that
$$
 \sup_{z \in (z_1, z_2)} \left[ -\frac{d_1^2(w_L')^2}{2w_L''} \right] < L
$$
for sufficiently large $L$. Assume the contrary for contradiction. Then there is a sequence $z_n \in (z_1, z_2)$, $L_n \to \infty$ such that
$$
  -\frac{d_1^2(w_L'(z_n))^2}{2w_L''(z_n)} > L_n,
$$
and
\begin{eqnarray}
	 \kappa w_L \geq d_2 z^2w_L'' - L_n + d_3 zw_L' + \left[ \delta w_L' - 1 - \log w_L'\right]. \label{eq:inequalityW}
\end{eqnarray}
Since $w_L \to w$ and both function are monotone and concave there exist constants $C_1, C_2$ such that
$$
 C_1 < w_L'(z) < C_2, \quad z \in [z_1, z_2]
$$
for all sufficiently large  $L$, and also $w_L'' \to 0$ as $n \to \infty$. But this contradicts (\ref{eq:inequalityW}) as $z_n$ takes values in a bounded interval so $w_L(z_n)$ is bounded as well.

 \textbf{Step 4.} We are going to show that there is a constant $K < 0$ which does not depend on $L$ such that
$$
  w_L''(z) < K, \quad z \in [z_1, z_2].
$$
Arguing again by contradiction suppose there is a sequence $z_n \in [z_1, z_2]$ such that $w_L''(z_n) \to \infty$.
Then analogously to Step 3, the right hand side of (\ref{eq:elliptic}) grows to infinity since $w_L'(z)$ on the interval that is bounded. At the same time the left hand side stays bounded as a value of a continuous function on a bounded interval.

 \textbf{Step 5.} Putting it all together, we have the following chain of implications. The functions $w_L$ are unique smooth solutions in the class of concave functions to the boundary problem (\ref{eq:elliptic}) for some sufficiently large $M > 0$. Since $w_L \to w$, it follows that $w$ is the unique viscosity solution of (\ref{eq:elliptic}) in the class of concave functions. On the other hand, the equation (\ref{eq:elliptic}) possesses the unique smooth solution, see \cite{krylov}, which must coincide with the viscosity solution. Thus $w$ is  a $C^2$-smooth function on $[z_1, z_2]$ and the claim of the theorem follows since the interval is arbitrary. $\bullet$
\end{proof}

\subsection{Asymptotic behavior of the value function.}
In this Section we examine the asymptotic behavior of the value function $V(t, l, h)$ and show that as $l/h \to \infty$ it becomes the classical Merton solution.
\begin{theorem}
	There is a positive constant $C_1$ such that
\begin{equation}
	M + \frac{\log(\kappa l)}{\kappa} \leq V(l, h) \leq M + \frac{\log(\kappa (l + C_1 \delta h))}{\kappa}, \label{thm:boundslog}
\end{equation}	
where
$$
M = \frac{1}{\kappa^2}\left[r + \frac{1}{2}\frac{(\alpha - r)^2}{\sigma^2} - \kappa \right]
$$
is a constant from the Merton's formula (\ref{eq:MertonSol}).
\end{theorem}
\begin{proof}
 The proof is based on the idea mentioned in Lemma \ref{PVF}, but in the specific exponentially distributed liquidation time case the bounds could be found explicitly. The left-hand inequality is obvious since any strategy $(\pi, c)$ for the classical problem with $L_0 = l, H_0 = 0$ is admissible for the problem with any non-zero initial endowment as well.
	For the right-hand side, let us consider a fictitious investment-consumption problem without any stochastic income but with an additional synthetic asset with the price process $S'$
\begin{eqnarray*}
	dS'_t &=& \alpha'S'_t + \sigma'_1S'_tdW_t, \quad t \geq 0 \\
	S'_0 &=& s', \quad s' > 0,
\end{eqnarray*}
with appropriate constants $\alpha'$ and $\sigma'$. Next, we define the \emph{initial wealth equivalent} of the stochastic income defined by
\begin{equation*}
	V_\delta(l, h) = \delta E_h\left[ \int^{\infty}_0 e^{-rt}\xi_t H_t dt \right],
\end{equation*}
where
$$
	\xi_t = \exp\left(-\frac{1}{2}(\theta^2_1 + \theta^2_2) + \theta_1W^{(1)}_t + \theta_2W_t\right),
$$
$\theta_1 = (\alpha-r)/\sigma_1$ and $\theta_2 = (\alpha' - r)/\sigma'_1$.

As we mention in the proof of Lemma \ref{PVF} (see page \pageref{PVF}), by a careful choice of the constants $\alpha', \sigma'$ the stochastic income rate $H_t$ can be replicated by a self-financing strategy on the complete market $(B_t, S_t, S'_t)$ with the additional initial endowment $f(h) < C_1 \delta h$, see \cite{KaratzasShreve}, \cite{HuangPages} and \cite{ZaripDuf}. Thus, any average  utility generated by the strategy $(\pi, c) \in \mathcal{A}(l, h)$ can be attained in the settings of a classical Merton's problem with the initial wealth $l + f(h) < l + C_1\delta h$. This actually gives the right-hand bound in (\ref{thm:boundslog}). $\bullet$
\end{proof}
 From this theorem we immediately get that $V(l, h)$ behaves as the classical Merton solution (\ref{eq:MertonSol}) as $\delta \to 0$ or $l/h \to \infty$.

\begin{corollary}
	\label{corr:d}
	$V_\delta(l, h)$ converges locally uniformly to $M + \log(\kappa l)/\kappa$ as $\delta \to 0$.
\end{corollary}

\begin{corollary}
	\label{lm:asymptV}
	$V(l, h) = M + \log(\kappa l)/\kappa + O(1/z)$ as $z = l/h \to \infty$. Also for the function $w(z)$ we obtain
	\begin{equation}
		w(z) = (M - K) + \frac{\log (\kappa z)}{\kappa} + O(1/z), \label{pr:aw}
	\end{equation}
\end{corollary}
\begin{proof}
	Indeed,
	$$
		\left|V(l, h) - M - \frac{\log(\kappa l)}{\kappa}\right| < \left|\frac{1}{\kappa} \left( \log(\kappa(l + \delta C_1 h)) - \log(\kappa l) \right) \right| =  O\left(\frac{1}{z}\right).
	$$
	The formula immediately follows from the form of $V(l, h)$.$\bullet$
\end{proof}

Finally, we verify that the optimal policies given by (\ref{eq:optimalCLog}) and (\ref{eq:optimalHLog}) asymptotically give the Merton strategy (\ref{eq:optimalCMerton}).

\begin{lemma}
\label{lm:asWP}
	For the value function $w(z)$ holds
	\begin{equation}
		w'(z) = \frac{1}{\kappa z} + o\left(\frac{1}{z}\right), \quad z \to \infty. \label{eq:asWP}
	\end{equation}
\end{lemma}
\begin{proof}
	Consider the function $w_{\lambda}$ defined as
	$$
		w_{\lambda}(z) = w(\lambda z) - \frac{\log(\lambda)}{\kappa},
	$$
so that $w_\lambda$ solves (\ref{eq:dualHJB}) but with the term
$$
F(w_z) = \max_{c \geq - \delta}\left[-cw_z + \log (c + \delta) \right]
$$
replaced by
$$
F_{\lambda}(w_z) = \max_{c \geq - \delta/\lambda}\left[-cw_z + \log (c + \frac{\delta}{\lambda}) \right].
$$
Then, by Corollary \ref{corr:d} $w_\lambda$ converges locally uniformly to the Merton's value function
$$
	v(z) = (M - K) + \frac{\log (\kappa z)}{\kappa}.
$$
We note that $v$ solves (\ref{eq:dualHJB}) with $\delta = 0$ that is delivered by
$$
F{_\infty}(\cdot) = \lim_{\lambda \to \infty} F_{\lambda}(\cdot).
$$

Thus, since $w_\lambda$ is concave, the uniform convergence of $w_\lambda$ to $v$ implies the convergence of derivatives, so
$$
	\lim_{\lambda \to \infty} {w'}_{\lambda}(z) = v'(z) = \frac{1}{\kappa z}.
$$
Hence,
$$
	\lim_{\lambda \to \infty} {w'}_{\lambda}(1) = \lim_{\lambda \to \infty} \lambda v'(\lambda) = \frac{1}{\kappa},
$$
which proves the lemma. $\bullet$
\end{proof}

\begin{theorem}
\label{thm:asPC}
	The following asymptotic formulae hold for the optimal policies (\ref{eq:optimalCLog}) and (\ref{eq:optimalHLog}) as $z = l/h \to \infty$.
	\begin{eqnarray}
		\frac{c^{\star}}{l} &\sim& \frac{1}{\kappa}, \label{lm:asC} \\
		\frac{\pi^{\star}}{l} &\sim& \frac{\alpha - r}{\sigma^2}. \label{lm:asP}
	\end{eqnarray}
\end{theorem}
\begin{proof}
	The relation (\ref{lm:asC}) immediately follows from Lemma~\ref{lm:asWP}. For the second part, we rewrite (\ref{eq:optimalHLog}) in a form
	$$
		\frac{\pi^{\star}}{l} = \frac{\eta \rho}{\sigma} - \frac{k_1}{\sigma^2}\frac{zv'(z)}{z^2v''(z)}.
	$$
	To calculate the limit value of $z^2v''(z)$ we rewrite (\ref{eq:logRed}) as a quadratic equation with respect to $w_{zz}$. Since $w_{zz} < 0$ we choose the negative root and obtain
	$$
		w''(z) = \frac{-B - \sqrt{B^2 - 4AC}}{2A},
	$$	
	where
	\begin{eqnarray}
		A = \frac{1}{2}\eta^2(1-\rho)^2z^2, ~~ B = k(zw') - 1 - (M - C)\kappa + o(1), ~~ C = -\frac{k_1^2}{2\sigma^2}(w')^2. \nonumber
	\end{eqnarray}
	Expanding all constants and using $zw' = 1/\kappa + o(1)$ we finally get
	$$
		z^2w''(z) = \frac{(\alpha - r)l}{\sigma^2} + o(1). \bullet
	$$
\end{proof}

The facts that the solution exists, is unique and smooth give an opportunity for numerical calculations. For example, basing on a script, developed by {\em Andersson, Svensson, Karlsson and Elias}, see \cite{goten}, with some modifications and corrections of minor mistakes we can obtain the solution for the exponential case and compare it with a two-dimensional Merton solution as shown on the Figure \ref{fig:Mer}.

\section{The case of Weibull distributed liquidation time and logarithmic utility function} \label{weibullsec}

One of the most natural ways to extend the framework of a randomly distributed liquidation time that we have described in the Section \ref{problem} is to introduce a distribution with a probability density function that has a local maximum unlike exponential distribution. It is very natural to expect that the assets of a certain type might have a time-lag between the moment when the sell offer is opened and a time when someone reacts on it. From the practitioner's point of view an empirical estimation of such time-lag is a natural measure of illiquidity that can give an insight into the strategy of a portfolio management.  In this Section we look closely on a Weibull distribution that has a local maximum. The Weibull distribution is commonly used in survival analysis, in reliability engineering and failure analysis, and in industrial engineering to describe manufacturing and delivery times. It seems to be quite adequate for the studied case. We demonstrate that the proposed framework is applicable for this case, show the existence and uniqueness of the solution and using a numerical algorithm generate an insight into how this case differs from the exponential illiquid and Merton's absolutely liquid cases.

In this Section we will discuss the case when the liquidation time $\tau$ is a random Weibull-distributed variable independent of the Brownian motions $(W^1, W^2)$.\\
The probability density function of the Weibull distribution is
	$$
		\phi(x, \lambda, k) = \begin{cases} \frac{k}{\lambda}\left(\frac{t}{\lambda}\right)^{k-1}e^{-(t/\lambda)^k}, & t \geq 0 \\ 0, & t < 0 \end{cases}.
	$$
	
	Let us also introduce as before the cumulative distribution function
	 \begin{equation} \label{WeibullPhi}
		\Phi(x, \lambda, k) = \begin{cases} 1 - e^{-(t/\lambda)^k}, & t \geq 0 \\ 0, & t < 0 \end{cases}
	\end{equation}
	and the survival function $\overline{\Phi}(t) = 1 - \Phi (t)$. We will often omit the constant parameters $\lambda$ and $k$ in notations for shortness.

 It is important to notice that when $k=1$ the Weibull-distribution turns into exponential one, that we have already discussed before and for $k>1$ its probability density has a local maximum. This situation corresponds to our economical motivation.\\
 The equation (\ref{eq:HJB21}) is the same as before but the term that corresponds to $\overline{\Phi}$ is naturally replaced by Weibull survival function
 \begin{eqnarray}
V_t (t, l, h) &+& \frac{1}{2}\eta^2h^2V_{hh} (t, l, h) + (rl + h)V_l (t, l, h) \nonumber \\
  &+& (\mu - \delta) hV_h (t, l, h) + \max_{\pi} G[\pi] + \max_{c \geq 0} H[c] = 0, \label{eq:HJBWei}\\
G[\pi] &=& \frac{1}{2}V_{ll}(t, l, h)\pi^2 \sigma^2 + V_{lh}(t, l, h)\eta\rho\pi\sigma h \nonumber \\
      &+& \pi(\alpha - r)V_l(t, l, h), \label{eq:GmaxWei} \\
H[c] &=& -cV_l (t, l, h) + e^{-(t/\lambda)^k} U(c), \label{eq:HmaxWei}
\end{eqnarray}

\begin{proposition}\label{propWei}
All the conditions of the Theorem \ref{MT} hold for the case of the Weibull distribution and, therefore, there exists a unique solution for the problem (\ref{eq:HJBWei}).
\end{proposition}

 Indeed the conditions $1.,$ $3.$ and $4.$ are not altered since we work with the same logarithmic utility and one can easily see that the cumulative function described in (\ref{WeibullPhi}) satisfies the condition $2.$ for the case $k>1$.

 Analogously to the equation (\ref{eq:HJBfinal}) one can obtain a two dimensional equation using a known reduction $z=l/h$. We study all the symmetry reductions of this model for the exponential and Weibull case in \cite{boya}. Yet here let us just list a two dimensional equation that corresponds to the Weibull case

 \begin{equation} \label{eq:HJBWei2}
 W_t - \frac{d_1^2}{2}\frac{(W_z)^2}{W_{zz}}  + d_2 z^2W_{zz} + d_3 z W_z + \delta W_z -  e^{-(t/\lambda)^k}\log W_z = 0,
\end{equation}
where $d_1, d_2$ and $d_3$ correspond to the constants for the general case (\ref{constd}).

 The function $\Psi_1(t) = \int_t^\infty \bPhi(s) ds$ can be defined explicitly as $\Psi_1(t)=\frac{\lambda}{k}\Gamma \left( \frac{1}{k},\left(\frac{t}{\lambda}\right)^k \right)$, where $\Gamma(\alpha, x)$ is an incomplete gamma function.
 For this function we can use the series representation by Laguerre polynomials and asymptotic representation \cite{absteg}, \cite{jel}.

 The lower bound for $W(z,t)$ can be found exactly as in Section \ref{bounds}
 $$ W(z, t) = V(t, l, h) - \Psi_1 \log h - \Psi_2(t) \geq \Psi_1(t) \log z + (\Theta(t) - \Psi_2(t)),$$
 where the behavior of the functions $\Psi_1,\Psi_2$ and $\Theta$ by $t \to \infty$ can be now well defined.

 The equation (\ref{eqnpsi2}) for the auxiliary function  $\Psi_2'(t)$ takes the form
 \begin{eqnarray} \label{eqnpsi2Weibull}
 \Psi_2'(t) + \lb(-\frac{\eta^2}{2} +(\mu-\delta)\rb)\frac{\lambda}{k}\Gamma \left( \frac{1}{k},\left(\frac{t}{\lambda}\right)^k \right) - e^{-(t/\lambda)^k}((t/\lambda)^k +1)=0 ,\\
 \Psi_2(t)\to 0, t\to \infty, \nonumber
 \end{eqnarray}
 The solution for this equation can be found explicitly
 \begin{equation} \label{psi2Weibull}
 \Psi_2 (t) = -\left(-\frac{\eta^2}{2} + (\mu - \delta) \right)\frac{\lambda}{k} \Gamma \left( \frac{1}{k}, \left(\frac{t}{\lambda} \right)^k \right) + e^{-\left(\frac{t}{\lambda}\right)^k} \left( \left(\frac{t}{\lambda} \right)^k +1 \right)
 \end{equation}
Equation (\ref{eqnpsi3}) for $\Theta$ is now
 \begin{eqnarray}
&&\Theta' +  \lb(r + \frac{1}{2}\frac{(\alpha-r)^2}{\sigma^2}\rb)\frac{\lambda}{k}\Gamma \left( \frac{1}{k},\left(\frac{t}{\lambda}\right)^k \right) \label{eqnpsi3Weibull}\\
&& -e^{-(t/\lambda)^k}\left(e^{-(t/\lambda)^k} +(t/\lambda)^k + \log \lambda - \log k + \log {\Gamma \left( \frac{1}{k},\left(\frac{t}{\lambda}\right)^k \right)}\right) = 0. \nonumber
\end{eqnarray}
And one can find an explicit solution for it as well
\begin{eqnarray} \label{psi3Weibull}
 \Theta(t) &=& -\left(r + \frac{(\alpha - r)^2}{2 \sigma^2} \right)\frac{\lambda}{k} \Gamma\left(\frac{1}{k}, \left(\frac{t}{\lambda} \right)^k \right) + \nonumber\\
&+& e^{-\left(\frac{t}{\lambda}\right)^k} \left( e^{-\left(\frac{t}{\lambda}\right)^k} + \left(\frac{t}{k} \right)^k + \ln \left( \frac{\lambda}{k} \Gamma\left(\frac{1}{k}, \left(\frac{t}{\lambda} \right)^k\right)\right) \right).
 \end{eqnarray}

Since $\frac{1}{k} > 0$ we can show that asymptotically as $t \to \infty $
\begin{eqnarray}
\Psi_1 (t) &\to& \frac{\lambda^k}{k} (t)^{1-k} e^{-(t/\lambda)^k}\left(1+O \left(t^{-k}\right)\right)\nonumber\\
\Psi_2 (t) &\to&  -\frac{1}{k} t e^{-(t/\lambda)^k} \left(1 + O \left(t^{-k}\right)\right), ~~k>1  ~~ \nonumber\\
\Theta (t) &\to&  \frac{\lambda - k}{\lambda k} t e^{-(t/\lambda)^k} \left(1 +  \frac{(k-1) k \lambda}{\lambda - k} t^{-k} \ln t +  O \left(t^{-k}\right)\right), ~~k>1  ~~  \nonumber
\end{eqnarray}

It follows from the asymptotic behavior that the value function in (\ref{eq:HJBWei2}) tends to zero faster than $e^{-\kappa t}$ and consequently Theorem \ref{MT} is applicable for the Weibull-distributed liquidation time.

On the Figure \ref{fig:Mer} one can see the results of the numerical simulation for consumption and investment strategies that we run for a Weibull and exponential case. As the parameter $k$ that is responsible for the form of Weibull distribution in (\ref{WeibullPhi}) increases the optimal policies differ significantly from the exponential liquidation-time case. As $z$ increases, i.e. the illiquid part of the portfolio becomes insufficiently small, we can see that all the policies tend to one solution which is, in fact, a Merton solution for a two-asset problem derived in \cite{Merton}.

\begin{figure}
\includegraphics[width=1\textwidth]{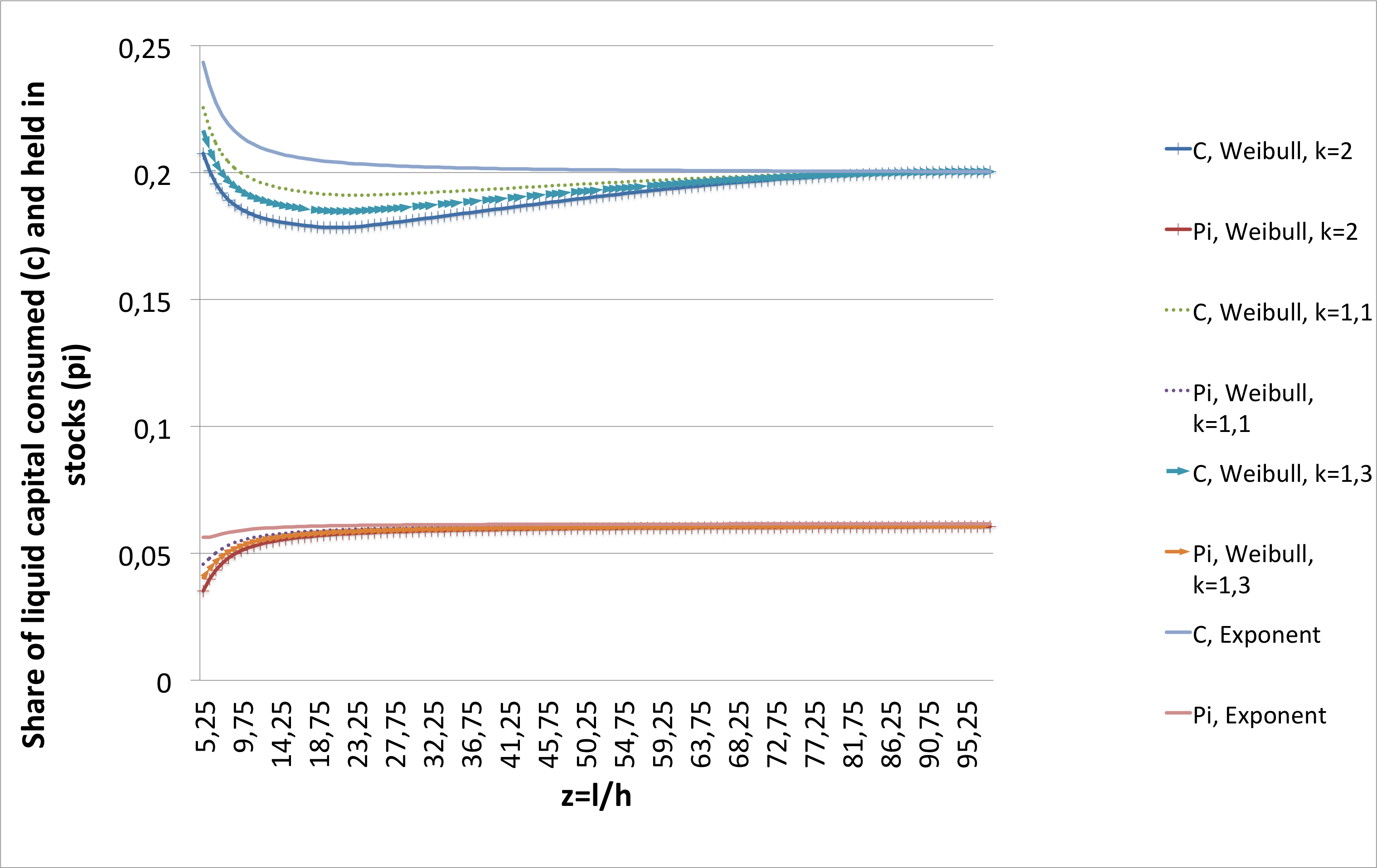}
\caption{Consumption stream $c$ and the share of liquid capital $\pi$ stored in a risky asset depending on the ratio between the liquid and illiquid asset. As illiquid asset value becomes infinitely small the policies tend to Merton policies for a two-asset problem. We used the following parameters for assets $r = 0.01,  \sigma = 0.5, \delta = 0.02, \rho =0.4, \mu = 0.05, \eta = 0.3$ and $\lambda = 2$ } \label{fig:Mer}
\end{figure}

It is especially important to note that the optimal policies significantly differ from Merton solution when illiquidity becomes higher. Already when an amount of illiquid asset is more than $5\%$ of the portfolio value the percentage of capital that is not invested in a risky stock is higher than in Merton model.

\section{Conclusion}

We have proposed a framework with which one can obtain a management strategy for a portfolio that consists of a liquid riskless and liquid risky assets and of an illiquid asset. We suppose that illiquid asset is liquidated in a random moment of time that has a prescribed distribution. In Theorem \ref{MT} we have proved the existence and uniqueness of the solution for a variety of the portfolio optimization problems. We have applied the obtained theorem to two different cases of exponential and Weibull liquidation time distributions. For the exponentially distributed random liquidation time but for existence and uniqueness we have proved the smoothness of the solution and found a lower and upper bound. For the Weibull distributed liquidation time with parameter $k>1$ we have demonstrated the applicability of a general Theorem \ref{MT} that proves the existence and uniqueness of a viscosity solution and also found a lower and upper bound for it. We have also demonstrated numerically that the resulting strategies for such portfolio differ from the Merton case yet tend to it when illiquidity becomes infinitely small.

% Author, Article title, Journal, Volume, page numbers (year)
% Author, Book title, page numbers. Publisher, place (year)

\end{document}